\documentclass[11pt]{article}

\usepackage[margin=1in]{geometry}
\usepackage{graphicx}
\usepackage[pdftex,colorlinks=true,linkcolor=blue,citecolor=blue,urlcolor=black]{hyperref}
\usepackage{amsmath, amsthm, amssymb}
\usepackage{subfigure}
\usepackage{comment}
\usepackage{url}
\usepackage{pdflscape}
\usepackage[ruled,lined,linesnumbered]{algorithm2e}
\usepackage{tikz}

\newcommand{\Z}{\mathbb{Z}}

\newcommand{\E}{\mathbb{E}}

\newcommand{\ket}[1]{| #1 \rangle}

\newcommand{\proj}[1]{| #1 \rangle \langle #1 |}

\DeclareMathOperator{\poly}{poly}
\DeclareMathOperator{\polylog}{polylog}

\newcommand{\be}{\begin{equation}}
\newcommand{\ee}{\end{equation}}
\newcommand{\bea}{\begin{eqnarray}}
\newcommand{\eea}{\end{eqnarray}}
\newcommand{\bes}{\begin{equation*}}
\newcommand{\ees}{\end{equation*}}
\newcommand{\beas}{\begin{eqnarray*}}
\newcommand{\eeas}{\end{eqnarray*}}

\makeatletter
\newtheorem*{rep@theorem}{\rep@title}
\newcommand{\newreptheorem}[2]{%
\newenvironment{rep#1}[1]{%
 \def\rep@title{#2 \ref{##1} (restated)}%
 \begin{rep@theorem}}%
 {\end{rep@theorem}}}
\makeatother

\newtheorem{thm}{Theorem}
\newtheorem*{thm*}{Theorem}

\newtheorem{lem}[thm]{Lemma}
\newtheorem*{lem*}{Lemma}
\newtheorem{prop}[thm]{Proposition}

\newtheorem{claim}[thm]{Claim}

\newreptheorem{thm}{Theorem}
\newreptheorem{lem}{Lemma}

\begin{document}


\title{Quantum pattern matching fast on average}
\author{Ashley Montanaro\thanks{Department of Computer Science, University of Bristol, UK; {\tt ashley@cs.bris.ac.uk}.}}
\maketitle

\begin{abstract}
The $d$-dimensional pattern matching problem is to find an occurrence of a pattern of length $m \times \dots \times m$ within a text of length $n \times \dots \times n$, with $n \ge m$. This task models various problems in text and image processing, among other application areas. This work describes a quantum algorithm which solves the pattern matching problem for random patterns and texts in time $\widetilde{O}((n/m)^{d/2} 2^{O(d^{3/2}\sqrt{\log m})})$. For large $m$ this is super-polynomially faster than the best possible classical algorithm, which requires time $\widetilde{\Omega}( n^{d/2} + (n/m)^d)$. The algorithm is based on the use of a quantum subroutine for finding hidden shifts in $d$ dimensions, which is a variant of algorithms proposed by Kuperberg.
\end{abstract}


\section{Introduction}

One of the most fundamental tasks in computer science is pattern matching: finding some desired data (the {\em pattern}) within a larger data set (the {\em text}). This problem has been of interest for decades, both in its own right and as part of more complicated questions in text processing, bioinformatics and image processing.

Here we consider the $d$-dimensional pattern matching problem, for arbitrary $d=O(1)$. Two examples of this problem are shown in Figure \ref{fig:pm}. We are given access to a text $T$ and a pattern $P$ over an alphabet $\Sigma$, with $|\Sigma| = q \ge 2$. Our task is to find an instance of $P$ within $T$, if such an instance exists. That is, writing $[n]:=\{0,\dots,n-1\}$ and thinking of $T$ and $P$ as functions $T:[n]^d \rightarrow \Sigma$, $P:[m]^d \rightarrow \Sigma$, we are required to output $s \in [n-m]^d$ such that $T(s+x) = P(x)$ for all $x \in [m]^d$, if such an $s$ exists; otherwise, we should output ``not found''. Throughout this work, we call any function of the form $S: [k]^d \rightarrow \Sigma$ a {\em string}, and think of  strings interchangeably as functions or $k \times \dots \times k$ arrays of elements of $\Sigma$. We assume throughout that $m \le n$.

The classical KMP algorithm of Knuth, Morris and Pratt~\cite{knuth77} from 1977 solves the pattern matching problem for $d=1$ in time $\Theta(n+m)$ in the worst case. This is clearly optimal, as every classical pattern-matching algorithm which is correct on all inputs must inspect every character of the pattern and the text. However, significantly improved runtimes can be achieved for more typical inputs. Consider a model where each character of the text is chosen at random from $\Sigma$, and the pattern is either uniformly random too (in which case, if it is long enough, it will not match the text with high probability), or is chosen to be a random substring of the text. A simple algorithm was given by Knuth~\cite[Section 8]{knuth77} which runs in time $O(n (\log_q m) / m + m)$ with high probability on such random inputs, while still running in time $O(n+m)$ in the worst case. Observe that the average-case runtime is substantially sublinear in $n$ for large $m$, but never better than $O(\sqrt{n \log n})$.

Shortly after this algorithm was developed, Yao proved an $\Omega((n/m) \log_q m)$ lower bound for the 1-dimensional matching problem, for random text and pattern~\cite{yao79b}. The bound extends to give a $\Omega((n/m)^d \log_q m)$ lower bound for the $d$-dimensional problem~\cite{karkkainen94}. More recently, an algorithm which runs in time $O((n/m)^d \log_q m + m^d)$ for the general $d$-dimensional problem, for random text and pattern, was given by K\"arkk\"ainen and Ukkonen~\cite{karkkainen94}. This is thus optimal up to the $O(m^d)$ term, which corresponds to preprocessing time for the pattern.

A quantum pattern-matching algorithm for the 1-dimensional case has been presented by Ramesh and Vinay~\cite{ramesh03}, which runs in time $\widetilde{O}(\sqrt{n})$ and hence achieves a square-root speedup over the best possible classical algorithm's worst-case complexity. However, the sublinear classical results mentioned above raise the following question: could there be a quantum pattern-matching algorithm which significantly outperforms its classical counterparts on average-case inputs which are more likely to occur in practice?

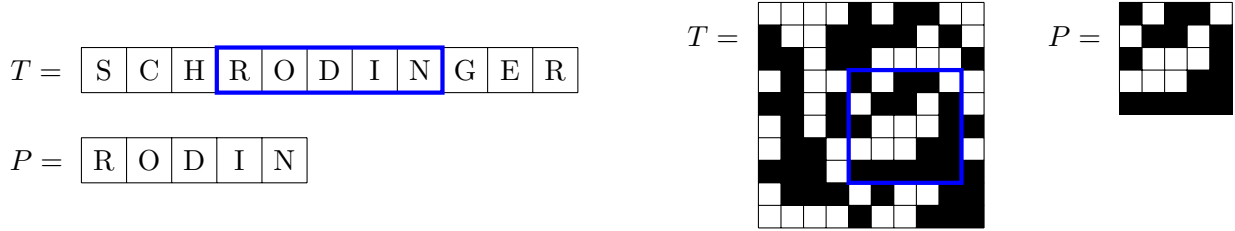
\begin{figure}
\begin{tikzpicture}[scale=0.6]
\draw (0,0) grid (11,1);
\draw (0,-2) grid (5,-1);
\node at (-1,0.5) {$T=$}; \node at (0.5,0.5) {S}; \node at (1.5,0.5) {C}; \node at (2.5,0.5) {H}; \node at (3.5,0.5) {R}; \node at (4.5,0.5) {O}; \node at (5.5,0.5) {D}; \node at (6.5,0.5) {I}; \node at (7.5,0.5) {N}; \node at (8.5,0.5) {G}; \node at (9.5,0.5) {E}; \node at (10.5,0.5) {R};
\draw[ultra thick,blue] (3,0) rectangle (8,1);
\node at (-1,-1.5) {$P=$}; \node at (0.5,-1.5) {R}; \node at (1.5,-1.5) {O}; \node at (2.5,-1.5) {D}; \node at (3.5,-1.5) {I}; \node at (4.5,-1.5) {N};
\begin{scope}[xshift=15cm,yshift=-3cm,scale=0.5]
\draw (0,0) grid (10,10);
\pgfmathsetseed{1};
\node at (-2,8.5) {$T=$};
\foreach \x in {0,...,9}{
\foreach \y in {0,...,9}{
\pgfmathifthenelse{random(0,1)==0}{"\noexpand\fill (\x,\y) rectangle (\x+1,\y+1);"}{}\pgfmathresult
}
}
\draw[ultra thick,blue] (4,2) rectangle (9,7);
\begin{scope}[xshift=2cm,yshift=3cm]
\pgfmathsetseed{1};
\node at (12,5.5) {$P=$};
\draw (14,3) grid (19,7);
\foreach \x in {0,...,9}{
\foreach \y in {0,...,9}{
\pgfmathifthenelse{(random(0,1)==0) && (4 <= \x) && (\x <= 8) && (2 <= \y) && (\y <= 6)}{"\noexpand\fill (\x+10,\y) rectangle (\x+11,\y+1);"}{}\pgfmathresult
}
}
\end{scope}
\end{scope}
\end{tikzpicture}
\caption{Examples of 1D and 2D pattern matching problems, with matches highlighted.}
\label{fig:pm}
\end{figure}


\subsection{Statement of results}
\label{sec:results}

We give a quantum algorithm which, for most instances of the $d$-dimensional pattern matching problem, is super-polynomially faster than the best possible classical algorithm.

\begin{thm}
\label{thm:random}
Assume $m = \omega(\log n)$. Let $T:[n]^d \rightarrow \Sigma$ be picked uniformly at random. Let $P : [m]^d \rightarrow \Sigma$ be picked either (a) by choosing an arbitrary $m \times \dots \times m$ substring of $T$, or (b) by choosing each element of $P$ uniformly at random from $\Sigma$. Then there is a quantum algorithm which runs in time $\widetilde{O}((n/m)^{d/2} 2^{O(d^{3/2}\sqrt{\log m})})$ and determines which is the case. In case (a), the algorithm also outputs the position at which $P$ matches $T$. The algorithm fails with probability $O(1/n^d)$, taken over both the choice of $T$ and $P$, and the algorithm's internal randomness.

Any classical bounded-error algorithm for the same problem must make $\widetilde{\Omega}( n^{d/2} + (n/m)^d )$ queries to $T$ and $P$ in total.
\end{thm}

The $\widetilde{O}$, $\widetilde{\Omega}$ notation suppresses factors logarithmic in $m$ and $n$ (see Propositions \ref{prop:random} and \ref{prop:classrandomlb} below for a more detailed statement of the quantum and classical complexities, respectively). The time complexity is stated in the standard quantum circuit model, assuming that a query to $T$ or $P$ uses time $O(1)$. We can think of $T$ and $P$ as either easily evaluated oracle functions in the query complexity model, or data stored in an efficiently accessible quantum random-access memory~\cite{giovannetti08}. All non-query operations performed by the algorithm contribute only polylogarithmic factors to the time complexity.

Observe that, for any fixed $d$, $2^{O(d^{3/2}\sqrt{\log m})} = o(m^\epsilon)$ for any $\epsilon > 0$. When $m$ is large, Theorem \ref{thm:random} thus demonstrates a super-polynomial separation between quantum and classical complexity (when $m$ is small, e.g.\ $O(\log n)$, straightforward use of Grover's algorithm is faster). For example, when $m = \Omega(n)$, we get a quantum algorithm running in time $\widetilde{O}(2^{O(d^{3/2}\sqrt{\log n})})$, as opposed to the best classical complexity of $\widetilde{\Omega}(n^{d/2})$. The omitted constants in the $O(d^{3/2}\sqrt{\log m})$ term in the exponent are not unreasonably high. For $d=1$, for example, the algorithm's runtime is $\widetilde{O}(\sqrt{n/m}\,2^{2.68\dots\sqrt{\log_2 m}})$. Theorem \ref{thm:random} is a rare example of a super-polynomial average-case separation between quantum and classical computation for a natural problem and a natural distribution on the input. An exponential average-case separation was previously proven~\cite{gavinsky11} for a related problem (an oracular hidden shift problem over $\Z_2^n$, see below), but that problem is arguably less natural than pattern matching.

Theorem \ref{thm:random} is based on a more general pattern matching result, which holds for non-random patterns and texts. In order to state this result more formally, we need some notation. For any string $S : [n]^d \rightarrow \Sigma$, we define a new string $S^{\triangleright k}:[n-k+1]^d \rightarrow \Sigma^{k^d}$, where $S^{\triangleright k}(s_1,\dots,s_d)$ is equal to the size $k \times \dots \times k$ substring of $S$ beginning at position $s_1,\dots,s_d$.  Formally, for any $f:[n]^d \rightarrow \Sigma$, $k \in \{1,\dots,n\}$, $s \in [n-k+1]^d$, let $f_{s,k}: [k]^d \rightarrow \Sigma$ be defined by $f_{s,k}(z_1,\dots,z_d) = f(s_1+z_1,\dots,s_n+z_n)$. Then
\[ S^{\triangleright k}(s_1,\dots,s_d) = S_{s,k}. \]
An example of this operation is shown in Figure \ref{fig:2dinj}. Note that we always consider $S^{\triangleright k}$ to be a string over the alphabet $\Sigma^{k^d}$; equivalently, a function $S^{\triangleright k}:[n-k+1]^d \rightarrow \Sigma^{k^d}$. Let the {\em injectivity length} of $S$, $\upsilon(S)$, be the minimal $k$ such that $S^{\triangleright k}$ is injective (i.e.\ all of its values are distinct).

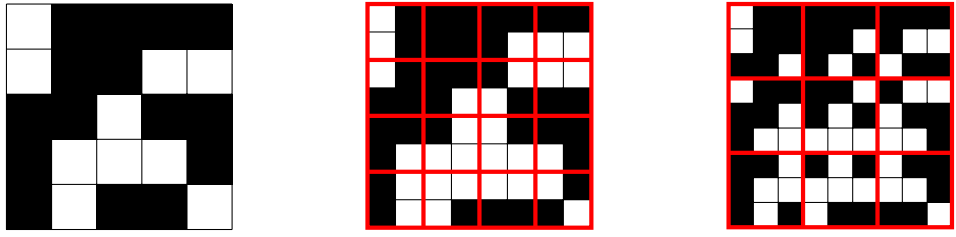
\begin{figure}
\begin{center}
\begin{tikzpicture}[scale=0.6]
\def\mytext{{{0,0,1,1,1},{1,1,1,0,0},{1,1,0,0,1},{1,0,1,0,1},{1,0,1,1,0}}}
\foreach \x in {0,...,4}{
\foreach \y in {0,...,4}{
\pgfmathifthenelse{\mytext[\x][\y]==1}{"\noexpand\fill (\x,-\y) rectangle (\x+1,-\y-1);"}{}\pgfmathresult
}
}
\draw (0,-5) grid (5,0);
\begin{scope}[xshift=8cm,scale=0.62]
\foreach \x in {0,...,3}{
\foreach \y in {0,...,3}{
\pgfmathparse{min(1,subtract(4,\x))}
\foreach \z in {0,...,\pgfmathresult} {
\pgfmathparse{min(1,subtract(4,\y))}
\foreach \w in {0,...,\pgfmathresult} {
\pgfmathifthenelse{\mytext[\x+\z][\y+\w]==1}{"\noexpand\fill (2*\x+\z,-2*\y-\w) rectangle (2*\x+\z+1,-2*\y-\w-1);"}{}\pgfmathresult
}
}
}
}
\draw (0,-8) grid (8,0);
\draw[red,ultra thick] (0,0) rectangle (8,-8);
\draw[red,ultra thick] (2,0) -- (2,-8); \draw[red,ultra thick] (4,0) -- (4,-8); \draw[red,ultra thick] (6,0) -- (6,-8);
\draw[red,ultra thick] (0,-2) -- (8,-2); \draw[red,ultra thick] (0,-4) -- (8,-4); \draw[red,ultra thick] (0,-6) -- (8,-6);
\end{scope}
\begin{scope}[xshift=16cm,scale=0.55]
\foreach \x in {0,...,2}{
\foreach \y in {0,...,2}{
\pgfmathparse{min(2,subtract(4,\x))}
\foreach \z in {0,...,\pgfmathresult} {
\pgfmathparse{min(2,subtract(4,\y))}
\foreach \w in {0,...,\pgfmathresult} {
\pgfmathifthenelse{\mytext[\x+\z][\y+\w]==1}{"\noexpand\fill (3*\x+\z,-3*\y-\w) rectangle (3*\x+\z+1,-3*\y-\w-1);"}{}\pgfmathresult
}
}
}
}
\draw (0,-9) grid (9,0);
\draw[red,ultra thick] (0,0) rectangle (9,-9);
\draw[red,ultra thick] (3,0) -- (3,-9);
\draw[red,ultra thick] (6,0) -- (6,-9);
\draw[red,ultra thick] (0,-3) -- (9,-3);
\draw[red,ultra thick] (0,-6) -- (9,-6);
\end{scope}
\end{tikzpicture}
\end{center}
\caption{Converting a non-injective 2D string $S$ (LHS) into an injective string $S^{\triangleright 3}$ (RHS). The centre string $S^{\triangleright 2}$ is not injective, so $\upsilon(S) = 3$. In addition, $\upsilon(S,2)=3$ as the $2 \times 2$ substring of $S^{\triangleright 2}$ in the middle of the top row is not injective.}
\label{fig:2dinj}
\end{figure}

We now consider the string $(S^{\triangleright k})_{s,m}:[m]^d \rightarrow \Sigma^{k^d}$ for $k \in \{1,\dots,n\}$, $m\in\{1,\dots,n-k-1\}$, $s \in [n-k-m+2]^d$. Define the $m$-{\em injectivity length} of $S$, $\upsilon(S,m)$, to be the minimal $k$ such that $(S^{\triangleright k})_{s,m}$ is injective for all $s$. Thus $\upsilon(S,m)\le k$ if every $m\times \dots \times m$ substring of $S^{\triangleright k}$ is injective. Observe that, for any $m$, $\upsilon(S,m) \le \upsilon(S) \le n$, but $\upsilon(S,m)$ can sometimes be much smaller than $\upsilon(S)$. For example, if $S : [n]^d \rightarrow \Sigma$ is constant, $\upsilon(S,1) = 1$, but $\upsilon(S) = n$.

Then the most general result we have is as follows:

\begin{thm}
\label{thm:main}
Fix $d = O(1)$. Let $T : [n]^d \rightarrow \Sigma$ and $P : [m]^d \rightarrow \Sigma$ satisfy $\upsilon(T,m),\upsilon(P) \le \nu \le m/2$, for some $\nu$. Further assume that, for every offset $s$ such that $P$ does not match $T$ at that offset, the fraction of positions $x \in [m]^d$ where $P(x) \neq T(x+s)$ is at least $\gamma$. Then there is a bounded-error quantum algorithm which outputs $s \in [m]^d$ such that $P$ matches $T$ at offset $s$, if such an $s$ exists; otherwise, the algorithm outputs ``not found''. The algorithm makes
\[ O\left( \left(\frac{n\,\log^2 m\,2^{\sqrt{(2\log_2 3)d\log_2 m}}}{m}\right)^{d/2} \left( \nu^d \log m\,2^{\sqrt{(2\log_2 3)d\log_2 m}} + \frac{1}{\sqrt{\gamma}} \right) \right) \]
queries to each of $T$ and $P$. The runtime is the same up to a $\polylog(m)$ factor.
\end{thm}

Theorem \ref{thm:main} may appear somewhat hard to digest. The intuition is that the algorithm is efficient, i.e.\ has runtime close to $O((n/m)^{d/2})$, when: the strings formed by concatenating all short substrings of both $T$ and $P$ are injective; and offsets where there is no match can be efficiently tested and discarded. The algorithm can thus be seen as achieving a speedup in a scenario somewhat similar to that considered in the field of property testing~\cite{montanaro13c}, where it has the promise that each potential match is either actually a match, or is far from being a match.

The complexity parameters of Theorem \ref{thm:main} are essentially optimal for, say, $\nu = 1$, up to terms of the form $2^{O(\sqrt{\log m})}$, which is $o(m^\epsilon)$ for any $\epsilon>0$. Indeed, it is easy to show the following bound using standard techniques: any bounded-error quantum algorithm which solves the problem described in Theorem \ref{thm:main} must make $\Omega((n/m)^{d/2}/\sqrt{\gamma} )$ queries in the case where $P$ is injective and completely known in advance, and $\upsilon(T,m)=1$. By comparison, any randomised classical pattern matching algorithm for the same problem must make $\Omega((n/m)^d/\gamma)$ queries. For the proof of these bounds, see Lemma \ref{lem:injlb1} below.


\subsection{Techniques}

Theorem \ref{thm:main} is ultimately based around the use of a quantum algorithm for finding hidden shifts in injective functions $f:\Z_{2^n}^d \rightarrow \Sigma$. The algorithm is a variant of algorithms of Kuperberg~\cite{kuperberg05}. Kuperberg's work described several algorithms: two for finding hidden shifts in injective functions $f:\Z_{2^n} \rightarrow \Sigma$, and one for finding hidden shifts in general abelian groups. The algorithm given here achieves essentially the same asymptotic complexity as the best algorithm given in~\cite{kuperberg05}, and appears somewhat simpler to analyse. In particular, we include a full proof of its correctness and complexity.

To use the algorithm, we first make the pattern and text injective. This is similar to the ``injectivisation'' idea used by Gharibi~\cite{gharibi13} in the context of quantum algorithms for abelian hidden shift problems, but here we need a slightly different notion, as used by Knuth~\cite{knuth77}, to ensure we preserve matching after injectivisation. We then apply the hidden shift algorithm by guessing an offset where the pattern matches the text. If our guess is fairly close, then the algorithm succeeds in finding the actual offset where the pattern matches. This guessing process is then wrapped within the use of the bounded-error variant of Grover's search algorithm~\cite{hoyer03} to obtain the final result. Theorem \ref{thm:random} is then derived by simply calculating the quantities $\nu$, $\gamma$ that occur in Theorem \ref{thm:main} for random strings.

Kuperberg showed in~\cite{kuperberg05} that, based on a similar idea of guessing offsets, his algorithms gave a super-polynomial quantum speedup for the task of finding an injective pattern of length $m$, promised to be hidden in an injective text of length $2m$. The contribution here is thus to generalise this idea to arbitrary dimensions $d>1$, to remove the restriction on the length of the text, and to relax the injectivity constraint. We also modify the promise that the pattern is guaranteed to be contained in the text to the promise that any non-matches can be tested efficiently. Observe that a constraint of this form is required if one seeks a runtime which is $o(m^{d/2})$. Imagine we are told an offset at which the pattern is claimed to match the text. If we have no lower bound on the number of positions at which it does not match the text if the claim is false, then verifying this claimed match would require $\Omega(m^{d/2})$ quantum queries in the worst case~\cite{bennett97}.


\subsection{Prior work}

Pattern matching is a fundamental algorithmic task, and has been studied in a number of different contexts.

\subsubsection{Pattern matching}

As well as the KMP algorithm already mentioned, another approach frequently used in practice classically is the Boyer-Moore algorithm~\cite{boyer77}, which achieves good performance for random inputs when the alphabet size is large. More recent classical work has pursued a number of other directions, such as approximate matching and search in compressed strings. For surveys of the (now vast) classical pattern matching literature, see for example~\cite{navarro01,crochemore07}.

Pattern matching has also been considered in the context of quantum computation. Grover's algorithm~\cite{grover97} can be seen as searching for a pattern of length 1 in a text of length $n$ using $O(\sqrt{n})$ queries to the text. The algorithm can be used na\"ively to find a pattern of length $m$ with a complexity of $O(\sqrt{nm})$ queries. However, Ramesh and Vinay~\cite{ramesh03} gave an improved algorithm which uses $\widetilde{O}(\sqrt{n})$ queries in the worst case. Their algorithm is based around the use of the powerful classical concept of deterministic sampling~\cite{vishkin91}. It is easy to see that this complexity is optimal in the worst case up to logarithmic factors, using standard lower bounds on the quantum query complexity of unstructured search~\cite{bennett97}. The algorithm of Ramesh and Vinay achieves a faster runtime than the algorithms described here in the case that $m$ is small (e.g.\ $O(\log n)$).

Curtis and Meyer describe an approach towards efficient quantum search for patterns (``templates'') within 2D images~\cite{curtis04}. This approach is based around the use of the quantum Fourier transform to compute correlations between the template and the image. As noted in~\cite{curtis04}, the algorithm is not complete and many details remain to be worked out before an efficient algorithm could be obtained.


\subsubsection{Hidden shifts}

There is now a fairly substantial body of work in the quantum setting on the closely related problem of finding hidden shifts. In an abstract setting, one is given access to two injective functions $f,g:G \rightarrow X$, for some abelian group $G$ and some set $X$, with the promise that $g(x) = f(x+s)$ for some $s \in G$, where $+$ is addition in the group $G$. The goal is to find $s$. It is known that this problem can be solved with only $O(\log |G|)$ quantum queries to $G$~\cite{ettinger04}. However, for certain groups $G$ it remains unknown whether there is a quantum algorithm which is similarly efficient with respect to time (i.e.\ runs in time $O(\polylog |G|)$), and this is considered to be a major open problem.

The case $G = \Z_n$ is the most relevant to our work here, which is equivalent to the hidden subgroup problem for the dihedral group~\cite{twamley00}. Algorithms to solve this problem have been given by Kuperberg~\cite{kuperberg05,kuperberg13} (whose work we will use and adapt below) and Regev~\cite{regev04}. These algorithms are all super-polynomially faster than the best possible classical algorithm for this problem: Kuperberg's algorithms run in time $2^{O(\sqrt{\log n})}$, while Regev's is almost as fast, running in time $2^{O(\sqrt{\log n \log \log n})}$. However, Kuperberg's algorithms use space $2^{O(\sqrt{\log n})}$, while Regev's algorithm uses space only $\poly(\log n)$. The more recent algorithm of~\cite{kuperberg13} uses only $O(\log n)$ quantum space, but $2^{O(\sqrt{\log n})}$ classical space. Our focus here is on optimising time complexity, so we base our algorithm on Kuperberg's.

The hidden shift problem has been studied for other groups $G$ too. Friedl et al.~\cite{friedl03} have given an efficient quantum algorithm running in time $\poly(\log |G|)$ for the case of $G = \Z_p^n$, where $p$ is a fixed prime and $n$ grows. A different generalisation, which can be seen as interpolating between the abelian hidden subgroup problem and the dihedral hidden subgroup problem, was studied by Childs and van Dam~\cite{childs07b}.

A number of works have studied a slightly different scenario in which one relaxes the injectivity constraint, but replaces it with complete knowledge of $f$. That is, one is given oracle access to a function $g:G \rightarrow X$, such that $g(x) = f(x+s)$ for some known function $f$, and is required to determine $s$. The complexity of the problem then depends on $f$. This problem was studied for $G = \Z_n$ for certain functions $f$ by van Dam, Hallgren and Ip~\cite{vandam06}, as well as by Moore et al.\ for prime $n$~\cite{moore07}. More recently, other works have considered the case $G=\Z_2^n$ in detail~\cite{roetteler09,roetteler10,gavinsky11,ozols13,childs13b}, characterising the complexity of the problem for many families of functions $f$.

To convert non-injective strings into injective strings, we concatenate adjacent symbols within the string. A similar idea was recently used by Gharibi to convert non-injective hidden shift problems over an arbitrary group into injective hidden shift problems~\cite{gharibi13}. The general framework takes a function $f:G \rightarrow X$ and a $k$-tuple $V \in G^k$, and defines a new function $f_V(x) = (f(x\cdot v_1),\dots,f(x\cdot v_k))$, where $\cdot$ is the group multiplication operation. In the hidden shift problem, if $g$ is equal to $f$ up to a shift (i.e.\ multiplication by an unknown group element $s$), then $g_V$ is equal to $f_V$ up to the same shift $s$. Gharibi showed that, for any choice of $V$ with $k = \Omega(\log |G|)$, if $f$ is picked at random then the probability that $f_V$ is not injective is low. This was then used to give an alternative proof that the quantum query complexity of the hidden shift problem over $\Z_2^n$ is low for most functions. The injectivisation procedure over the group $\Z_n$ used here is slightly different: in order to preserve the property of the pattern matching the text, we do not allow the shifts to wrap around and only consider the set $V = \{1,\dots,k\}$.

The hidden shift problem over $\Z_n$ has also been studied classically. In particular, Andoni et al.~\cite{andoni13} have considered a noisy variant where one has access to two boolean functions $f,g:\Z_n \rightarrow \{0,1\}$ such that $g(x) = f(x + s) + r$ for some shift $s$, and some string $r \in \{0,1\}^n$ of random bits where the probability that each bit of $r$ is equal to 1 is independent and equal to $\eta$, for some fixed constant $\eta$. The goal is again to find $s$. A practical motivation for this problem comes from GPS synchronisation. In the case where the values taken by $f$ are uniformly random, it is shown in~\cite{andoni13} that the problem can be solved in sublinear time; their algorithm runs in time $O(n^{0.641})$. This algorithm is not that far from optimal, as the hidden shift problem over $\Z_n$ has a lower bound of $\Omega(\sqrt{n})$ queries~\cite{batu03}.


\subsection{Organisation}

We begin, in Section \ref{sec:pmshift}, by describing how a quantum algorithm for the hidden shift problem can be used to obtain a general pattern-matching algorithm (Theorem \ref{thm:main}). Section \ref{sec:random} contains the calculations showing that this can be applied to random strings (Theorem \ref{thm:random}). In Section \ref{sec:lower}, we prove the required classical lower bounds to complete the proof of Theorem \ref{thm:random}, and also prove quantum lower bounds showing that our algorithms are not too far from optimal. For completeness, in this section we also give a classical algorithm matching the classical lower bound. Section \ref{sec:dihedral} describes the quantum algorithm for the hidden shift problem. We conclude in Section \ref{sec:outlook}.



\section{Quantum pattern matching based on finding hidden shifts}
\label{sec:pmshift}

We will need the following simple lemma, whose proof follows immediately from amplitude amplification~\cite{brassard02}.

\begin{lem}
Assume we have query access to $S,T:[m] \rightarrow \Sigma$ such that either $S=T$, or $|\{x\mid S(x)\neq T(x)\}| \ge \gamma m$, for some $0 < \gamma \le 1$. Then there is a quantum algorithm {\sf Check} such that: in the first case, {\sf Check} accepts with certainty; in the second case, {\sf Check} rejects with probability at least $2/3$; {\sf Check} makes $O(1/\sqrt{\gamma})$ queries. The runtime is the same up to a $\polylog (m)$ factor.
\end{lem}

The technical core of our algorithm is the following result:

\begin{thm}
\label{thm:dmatchapprox}
Let $d=O(1)$ and let $X$ be an arbitrary finite set. Let $f:\Z_{2^n}^d \rightarrow X$ and $g:\Z_{2^n}^d \rightarrow X$ be injective functions such that $\Pr_x[g(x) \neq f(x+s)] = O(n^{-2} 2^{-\sqrt{(2\log_2 3)dn}})$ for some $s \in \Z_{2^n}^d$. Then there is a quantum algorithm which outputs $s$ with bounded error using $O(n 2^{\sqrt{(2\log_2 3)dn}}) = O(n 2^{1.781\dots\sqrt{dn}})$ queries to each of these functions. The runtime is the same up to a $\poly(n)$ factor.
\end{thm}

We prove Theorem \ref{thm:dmatchapprox} later, in Section \ref{sec:dihedral}. We show here that it implies Theorem \ref{thm:main}, which we restate for convenience.

\begin{repthm}{thm:main}
Fix $d = O(1)$. Let $T : [n]^d \rightarrow \Sigma$ and $P : [m]^d \rightarrow \Sigma$ satisfy $\upsilon(T,m),\upsilon(P) \le \nu \le m/2$, for some $\nu$. Further assume that, for every offset $s$ such that $P$ does not match $T$ at that offset, the fraction of positions $x \in [m]^d$ where $P(x) \neq T(x+s)$ is at least $\gamma$. Then there is a bounded-error quantum algorithm which outputs $s \in [m]^d$ such that $P$ matches $T$ at offset $s$, if such an $s$ exists; otherwise, the algorithm outputs ``not found''. The algorithm makes
\[ O\left( \left(\frac{n\,\log^2 m\,2^{\sqrt{(2\log_2 3)d\log_2 m}}}{m}\right)^{d/2} \left( \nu^d \log m\,2^{\sqrt{(2\log_2 3)d\log_2 m}} + \frac{1}{\sqrt{\gamma}} \right) \right) \]
queries to each of $T$ and $P$. The runtime is the same up to a $\polylog(m)$ factor.
\end{repthm}

\begin{proof}
Let $m'$ be the largest power of 2 less than or equal to $m - \nu$. The algorithm is based on the following procedure {\sf RoughCheck}, which takes as input a shift $t \in [n-\nu - m'+ 2]^d$ and a tolerance $\epsilon \in (0,1]$, and is designed to accept if $t$ is quite close to a position $t'$ at which $P$ matches $T$:
\begin{enumerate}
\item Apply the algorithm of Theorem \ref{thm:dmatchapprox} to $T' := (T^{\triangleright \nu})_{t,m'}$ and $P' := (P^{\triangleright \nu})_{0,m'}$. Let $\ell \in [m']^d$ be the offset output by the algorithm. If there exists $i$ such that $\ell_i > \epsilon m'$, reject.
\item Otherwise, apply {\sf Check} to $T_{t+\ell,m}$ and $P$, and accept if and only if it accepts.
\end{enumerate}
In this definition, the notation is as used in Section \ref{sec:results}: thus $(T^{\triangleright \nu})_{t,m'}$ denotes the $m' \times \dots \times m'$ substring of $T^{\triangleright \nu}$ starting at offset $t$, and $(P^{\triangleright \nu})_{0,m'}$ is the first $m' \times \dots \times m'$ characters of $P^{\triangleright \nu}$. Note that $T'$ and $P'$ are injective, so Theorem \ref{thm:dmatchapprox} can indeed be applied to them. It is immediate that {\sf RoughCheck} uses $O(\nu^d \log m 2^{\sqrt{(2\log_2 3)d\log_2 m}} + 1/\sqrt{\gamma})$ queries, where the $O(\nu^d)$ term comes from simulating a query to $T^{\triangleright \nu}$ using $\nu^d$ queries to $T$. We now show that, for $\epsilon = O((\log^{-2}_2 m')2^{-\sqrt{(2\log_2 3)d\log_2 m'}})$, {\sf RoughCheck} is a bounded-error verifier for the property of $P$ matching $T$ at some offset $t' \in [n-m+1]^d$, where $t_i \le t'_i \le t_i+ \epsilon m'$ for all $i \in \{1,\dots,d\}$. Call this property {\em $\epsilon$-matching}. First assume there does exist such an offset $t'$. Then $P^{\triangleright \nu}$ also matches $T^{\triangleright \nu}$ at the same offset. Taking addition modulo $m'$ in each dimension,
\[ |\{x \in [m']^d :P'(x) \neq T'(x+t'-t)\}| \le \sum_{i=1}^d (t'_i - t_i)(m')^{d-1} = O((m')^d \epsilon), \]
where the first inequality is a rough bound on the size of the complement of one $d$-dimensional cube within another, in terms of $(d-1)$-dimensional slices. The algorithm of Theorem \ref{thm:dmatchapprox} therefore outputs $\ell = t'-t$ with bounded failure probability, and if it does so then {\sf Check} accepts with certainty given the two strings $T_{t',m}$ and $P$. On the other hand, if there is no such offset $t'$, there are two possibilities: the algorithm of Theorem \ref{thm:dmatchapprox} could output a correct match between $P$ and $T$, but at an offset $t'$ which fails to satisfy $t_i \le t'_i \le t_i+ \epsilon m'$ for at least one $i$; or otherwise, the algorithm could output an incorrect claimed match. In the former case, this will be detected by the check in step 1. In the latter case, {\sf Check} will detect this with bounded error. We therefore see that, if $P$ $\epsilon$-matches $T$ at offset $t$, {\sf RoughCheck} accepts except with bounded failure probability; whereas if $P$ does not $\epsilon$-match $T$ at offset $t$, {\sf RoughCheck} rejects except with bounded failure probability.

If we pick $t$ at random from $[n-\nu - m'+ 2]^d$, and $P$ matches $T$ at at least one offset, the probability that $P$ matches $T$ at some offset $t'$, where $t_i \le t'_i \le t_i+ \epsilon m'$, is at least $(\epsilon m'/n)^d$. Applying the bounded-error version of Grover's search algorithm~\cite{hoyer03}, we can find a position at which $P$ $\epsilon$-matches $T$ with $O((n/(\epsilon m'))^{d/2})$ uses of {\sf RoughCheck}, and again with bounded failure probability. Once such a position is found, one more use of the algorithm of Theorem \ref{thm:dmatchapprox} suffices to output the offset within this range at which $P$ matches $T$. The claimed result for the number of queries follows, substituting the value of $\epsilon$ back in and using $m' = \Omega(m)$. The argument for the runtime bound is similar.
\end{proof}

Note that there is also an algorithm which does not need to know the value of $\nu$ in advance, at the expense of a small additional runtime factor. We simply run the algorithm of Theorem \ref{thm:main} multiple times, doubling a guess for $\nu$ each time. Each time that the algorithm claims we have a match, we can use {\sf Check} to determine if it really is a match. To achieve a sufficiently small probability of failure, we need to repeat {\sf Check} at most $O(\log \nu)$ times. We can also get an algorithm with no dependence on $\gamma$ if we make some slightly different assumptions (cf.\ Lemma \ref{lem:injlb1} below for why these assumptions are necessary).

\begin{thm}
\label{thm:main2}
Let $T : [n]^d \rightarrow \Sigma$ and $P \in [m]^d \rightarrow \Sigma$ satisfy $\upsilon(T),\upsilon(P) \le \nu \le m/2$, for some $\nu$. Further assume that $P$ matches $T$ at some position $i$ (necessarily unique). Then there is a bounded-error quantum algorithm which outputs $i$ and makes
\[ O\left( \left(\frac{n}{m}\right)^{d/2} \nu^d \log^{d+1} m\,2^{\sqrt{(2\log_2 3)d\log_2 m}(d/2+1)} \right) \]
queries to each of $T$ and $P$. The runtime is the same up to a $\polylog(m)$ factor.
\end{thm}

\begin{proof}
Let $m'$ be the largest power of 2 less than or equal to $m - \nu$. The algorithm is based on the following procedure {\sf RoughCheck2}, which takes as input a shift $t \in [n-\nu - m'+ 2]^d$ and a tolerance $\epsilon \in (0,1]$:
\begin{enumerate}
\item Apply the algorithm of Theorem \ref{thm:dmatchapprox} to $T' := (T^{\triangleright \nu})_{t,m'}$ and $P' := (P^{\triangleright \nu})_{0,m'}$. Let $\ell$ be the offset output by the algorithm. If there exists $i$ such that $\ell_i > \epsilon m'$, reject.
\item If $T^{\triangleright \nu}(t + \ell) = P^{\triangleright \nu}(0)$, output $\ell$. Otherwise, reject.
\end{enumerate}
The analysis is the same as for Theorem \ref{thm:main}, replacing {\sf RoughCheck} with {\sf RoughCheck2}, which uses $O(\nu^d \log m 2^{\sqrt{(2\log_2 3)d\log_2 m}})$ queries (and has no dependence on $\gamma$). The key difference is that, if $T^{\triangleright \nu}(t + \ell) \neq P^{\triangleright \nu}(0)$, we can be sure that $t + \ell$ is not the position where $P$ matches $T$. This follows from injectivity of $T^{\triangleright \nu}$ and $P^{\triangleright \nu}$ and the fact that $P$ is indeed contained somewhere within $T$.
\end{proof}


\section{Pattern matching in random strings}
\label{sec:random}

We now show that Theorem \ref{thm:main} can be applied to random patterns and texts. This simply involves calculating the parameters required to apply the theorem. First we show that random strings can be made injective by only considering a small number of subsequent positions. A similar result was previously shown in a more general context by Gharibi~\cite{gharibi13}, using a different notion of injectivisation. Recall that we write $q = |\Sigma|$.

\begin{lem}
\label{lem:randominj}
Let $S : [n]^d \rightarrow \Sigma$ be uniformly random. Then $\Pr[\upsilon(S) \ge (3d \log_q n)^{1/d}] \le 1/n^d$.
\end{lem}

\begin{proof}
Consider the string $S^{\triangleright k}$ for arbitrary $k$. For any offset $s \in [n-k]^d$ and non-zero $\delta \in [n]^d$, the probability over $S$ that $S^{\triangleright k}(s) = S^{\triangleright k}(s+\delta)$ is
\[ \Pr_S\left[\bigwedge_{t \in s + [k]^d} (S(t) = S(t+\delta))\right]. \]
This probability is exactly $q^{-k^d}$, whether or not there exist a pair $t$, $t+\delta$ that are both contained in $s + [k]^d$. To see why, observe that we can think of choosing $S$ by fixing its values $S(x)$ one by one, in some arbitrary order such that if $x_i \le y_i$ for all $i \in \{1,\dots,d\}$, $S(x)$ is fixed before $S(y)$. Then, however the value $S(t)$ was chosen, the value of $S(t+\delta)$ is uniformly random. So $\Pr_S[S(t) = S(t+\delta)] = 1/q$. Taking the union bound over all possible choices of $s$ and $\delta$, we obtain an upper bound of $n^{2d} q^{-k^d}$ on the probability that $S^{\triangleright k}$ is non-injective. Choosing $k = \lceil (3d \log_q n)^{1/d} \rceil$ gives the claimed bound.
\end{proof}

We can also show that random strings are unlikely to be close to matching at any offset.

\begin{lem}
\label{lem:randommatch}
Let $T : [n]^d \rightarrow \Sigma$ be arbitrary, and let $P:[m]^d \rightarrow \Sigma$ be uniformly random. Then the probability that there exists an offset $s \in [n-m]^d$ such that $|\{t : P(t) = T(s+t)\}| \ge (3/4)m^d$ is at most $n^d e^{-m^d/8}$.
\end{lem}

\begin{proof}
For a fixed offset $s$, using a Chernoff bound we have
\[ \Pr_P [ |\{t : P(t) = T(s+t) \}| \ge m^d/q + \delta ] \le e^{-2\delta^2/m^d}. \]
Taking a union bound over all offsets $s$,
\[ \Pr_P [ \exists s, |\{t : P(t) = T(s+t) \}| \ge m^d/q + \delta ] \le n^d e^{-2\delta^2/m^d}. \]
Using $q \ge 2$ and taking $\delta = m^d/4$,
\[ \Pr_P [ \exists s, |\{t : P(t) = T(s+t)\}| \ge (3/4)m^d ] \le n^d e^{-m^d/8}. \]
\end{proof}

If $m$ satisfies $m^d \ge 16d \ln n$, the bound obtained from Lemma \ref{lem:randommatch} is at most $1/n^d$. Note that a meaningful bound cannot be found for $m$ significantly smaller than this. If we take a random text $T:[n]^d \rightarrow \Sigma$, and divide it up into $(n/m)^d$ blocks of size $m^d$, the probability that an arbitrary pattern fails to match a given block is $1-q^{-m^d}$, so the probability that it fails to match all blocks is
\[ (1 - q^{-m^d})^{(n/m)^d} \le e^{- n^d /(m^d q^{m^d}) }. \]
This probability is small for $m^d = O(\log n)$. That is, if the pattern is too short, it is likely to ``unintentionally'' match the text somewhere.

Combining Lemmas \ref{lem:randominj} and \ref{lem:randommatch}, and inserting these parameters into Theorem \ref{thm:main}, we get the following result.

\begin{prop}
\label{prop:random}
Let $m = \omega(\log n)$ and fix $d=O(1)$. Let $T:[n]^d \rightarrow \Sigma$ be picked uniformly at random. Let $P : [m]^d \rightarrow \Sigma$ be picked either (a) by choosing an arbitrary $m \times \dots \times m$ substring of $T$, or (b) by choosing each element of $P$ uniformly at random from $\Sigma$. Then there is a quantum algorithm which makes
\[ O\left( \left(\frac{n}{m}\right)^{d/2} 2^{(d/2+1) \sqrt{(2\log_2 3)d\log_2 m}} \log^{d+1} m \log n \right) \]
queries to $T$ and $P$ and determines which is the case. The runtime is the same up to a $\polylog(m)$ factor. The algorithm fails with probability $O(1/n^d)$ over the choice of $T$ and $P$, and with an arbitrarily small probability over its own internal randomness. In case (a), the algorithm also outputs the position at which $P$ matches $T$.
\end{prop}

This is the first part of Theorem \ref{thm:random}.


\section{Lower bounds and classical upper bounds}
\label{sec:lower}

We now prove nearly matching quantum and classical lower bounds. We will actually lower-bound the complexity of the following variant of the pattern matching problem. We are given access to a pattern $P:[m]^d \rightarrow \Sigma$, and a text $T:[n]^d \rightarrow \Sigma$. We are promised that either there is a unique offset at which $P$ matches $T$, or there is no such offset. Our task is to determine which is the case. Call this the {\em pattern detection} problem. As with the quantum upper bounds discussed earlier, we impose the additional promise that, for every offset $s$ such that $P$ does not match $T$ at that offset, the fraction of positions $x \in [m]^d$ where $P(x) \neq T(x+s)$ is at least $\gamma$.

To prove a quantum lower bound for this problem, we will use the following result of Ambainis~\cite{ambainis02}.

\begin{thm}[Ambainis~\cite{ambainis02}]
\label{thm:adversary}
For any $f:S \subseteq \{0,1\}^n \rightarrow \{0,1\}$, let $X,Y \subseteq S$ be two sets of inputs such that $f(x) \neq f(y)$ for all $x \in X$ and $y \in Y$. Further let $R \subseteq X \times Y$ be such that
\begin{enumerate}
\item For every $x \in X$, there exist at least $\mu$ different $y \in Y$ such that $(x,y) \in R$.
\item For every $y \in Y$, there exist at least $\mu'$ different $x \in X$ such that $(x,y) \in R$.
\item For every $x \in X$ and $i \in \{1,\dots,n\}$, there are at most $\lambda$ different $y \in Y$ such that $(x,y) \in R$ and $x_i \neq y_i$.
\item For every $y \in Y$ and $i \in \{1,\dots,n\}$, there are at most $\lambda'$ different $x \in X$ such that $(x,y) \in R$ and $x_i \neq y_i$.
\end{enumerate}
Then any quantum algorithm computing $f(x)$ with bounded failure probability for all $x \in S$ uses $\Omega\left(\sqrt{\frac{\mu\mu'}{\lambda\lambda'}}\right)$ queries to the bits of $x$.
\end{thm}

First we show that the quantum pattern matching algorithm given here is not far from optimal, and give a corresponding classical lower bound.

\begin{lem}
\label{lem:injlb1}
Any bounded-error quantum algorithm for the pattern detection problem must make $\Omega((n/m)^{d/2}/\sqrt{\gamma} )$ queries, even if $P$ is injective, $\upsilon(T,m)=1$, and $P$ is completely known in advance. Any randomised classical pattern matching algorithm for the same problem must make $\Omega((n/m)^d/\gamma)$ queries.
\end{lem}

\begin{proof}
Both the quantum and classical lower bounds are based on the same hard input distribution. We use the alphabet $\Sigma = [2m]^d$. Set $P(x_1,\dots,x_d) = (x_1,\dots,x_d)$ and fix $n = mp$ for some integer $p$. Divide $T$ into $p^d$ blocks of size $m \times \dots \times m$, with each block initially being equal to $P$. Within each block, either change an arbitrary $\gamma$ fraction of the elements of $T$ by adding $m$ to each component of the sequence, or change none of them. Then $P$ only matches $T$ at one or more position at offsets given by blocks, and within each block except one fails to match $T$ at a $\gamma$ fraction of positions. The only information given by querying elements of the text is whether that element is equal to the corresponding element of $P$, or not. We can therefore think of the text as an $n^d$-bit string, which is divided into $p^d$ blocks of size $m \times \dots \times m$; within each block the string either takes the value 1 at a $\gamma$ fraction of positions, or at no positions. The goal is to determine whether there exists a block where the string is equal to 0 at all positions. This is equivalent to evaluating a 2-level OR-AND tree with a promise on the number of 1's in each block.

For the quantum lower bound, we now apply Theorem \ref{thm:adversary}. Let $X$ be the set of all bit-strings with exactly one block containing only 0's, and all other blocks containing $\gamma m^d$ 1's; and let $Y$ be the set of all bit-strings with all blocks containing $\gamma m^d$ 1's. Finally let $R$ be the set of all pairs $(x,y) \in X \times Y$ such that $x$ and $y$ only differ within exactly one block. Then one can readily calculate, in the notation of Theorem \ref{thm:adversary}, that $\mu = \binom{m^d}{\gamma m^d}$, $\mu' = p^d$, $\lambda = \binom{m^d-1}{\gamma m^d-1}$, $\lambda' = 1$. We therefore obtain a lower bound of $\Omega(\sqrt{\mu\mu'/(\lambda\lambda')}) = \Omega(\sqrt{p^d/\gamma}) = \Omega((n/m)^{d/2}/\sqrt{\gamma})$ queries.

For the classical lower bound, we use the Yao principle that it suffices to prove a lower bound on deterministic algorithms which succeed on most inputs picked from some probability distribution (here, the distribution described above). Imagine that the output should be 0. The algorithm cannot be confident that this is the case until it has seen a 1 in each of the $p^d$ blocks. But, within each block, the number of queries required to do so is $\Omega(1/\gamma)$. Multiplying these bounds gives the claimed result.
\end{proof}

If we have no a priori limitation on $\gamma$, it can be as low as $1/m^d$, so in this case we have a bound of $\Omega(\sqrt{n})$ for $d=1$. This is achieved, up to polylogarithmic factors, by the algorithm of Ramesh and Vinay~\cite{ramesh03}.

We now give a general classical lower bound for the pattern detection problem for the case of injective functions.

\begin{lem}
\label{lem:injclasslb}
Any bounded-error classical algorithm which solves the pattern detection problem for all injective functions $P:[m]^d \rightarrow \Sigma$ and $T:[n]^d \rightarrow \Sigma$ must make at least $\Omega(n^{d/2} + (n/m)^d + 1/\gamma)$ queries in total.
\end{lem}

Part of the bound achieved by Lemma \ref{lem:injclasslb} is similar to a result of Batu et al.~\cite{batu03}, who proved an $\Omega(\sqrt{n})$ lower bound for the hidden shift problem where $f,g:\Z_n \rightarrow \Z$ and we are promised that $g(x) = f(x+s)$ for some $s \in \Z_n$.

\begin{proof}
We first observe that lower bounds of $\Omega((n/m)^d)$ and $\Omega(1/\gamma)$ are easy. For the former, we divide the text into  $(n/m)^d$ contiguous $m \times \dots \times m$ blocks, and impose the promise that the pattern matches the text within one of the blocks. Then the number of queries used by any classical algorithm which identifies which block this is must be lower bounded by a quantity proportional to the number of blocks. For the latter, if we promise that the pattern either matches the text at some known offset, or has a $\gamma$ fraction of its entries not matching, to determine which is the case requires $\Omega(1/\gamma)$ queries to the text.

So it remains to prove an $\Omega(n^{d/2})$ lower bound. For the proof we use the alphabet $\Sigma = [n^d]$. Consider two distributions $\mathcal{D}_0$, $\mathcal{D}_1$. In the first distribution, $T$ is a uniformly random permutation of $[n^d]$, and $P$ is formed by choosing $m$ integers from $[n^d]$ at random and then randomly permuting them. In the second distribution, $T$ is formed in the same way, and $P$ is formed by taking a random sub-block within $T$ of size $m \times \dots \times m$. We show that any deterministic classical algorithm making $o(n^{d/2})$ queries cannot distinguish between $\mathcal{D}_0$ and $\mathcal{D}_1$. By the Yao principle, this suffices to prove the corresponding bound for randomised algorithms.

Imagine the algorithm makes $K$ queries to $T$ and $L$ queries to $P$. Define a collision to be the result of some query to $T$ which equals the result of some previous query to $P$, or vice versa. Then if the classical algorithm has not seen any collisions when it terminates, the distributions $\mathcal{D}_0$ and $\mathcal{D}_1$, conditioned on the query results, are indistinguishable. It therefore suffices to upper-bound the probability that the algorithm finds a collision. Each result of a query to $P$ or $T$ that is not a collision gives no additional information about $P$ or $T$. Therefore, until at least one collision is found, we can assume that the choice of subsequent queries does not depend on previous queries. In particular, we can assume that all the queries to $P$ are made first. Then, using a union bound, the probability that one of the queries to $T$ is a collision is at most
\[ \frac{L}{n^d} + \frac{L}{n^d-1} + \dots + \frac{L}{n^d-K+1}. \]
Assuming that $K \le n^d/2$, this is at most $2KL/n^d$. Therefore, to see a collision with probability at least $1/2$, we need $KL = \Omega(n^d)$, implying $K + L = \Omega(n^{d/2})$.
\end{proof}

For completeness, we describe a classical algorithm which matches the bound of Lemma \ref{lem:injclasslb}.

\begin{thm}
\label{thm:classub}
Let $P:[m]^d\rightarrow \Sigma$ and $T:[n]^d \rightarrow \Sigma$ be injective. Then there is a bounded-error classical algorithm which finds $P$ within $T$, if it exists, using $O(n^{d/2} + (n/m)^d + 1/\gamma)$ queries.
\end{thm}

\begin{proof}
The algorithm proceeds as follows. Read in the $k^d$ elements of the $k \times \dots \times k$ substring $P_{0,k}$, for some $k$ to be determined. Then divide the text into $O((n/k)^d)$ contiguous blocks of size at most $k \times \dots \times k$ and read the entry of the text at position $(0,\dots,0)$ within each block. If $P$ is contained within $T$, exactly one of these characters will match one of the characters previously read from $P$. Once a matching character has been found, the algorithm samples $O(1/\gamma)$ elements from the pattern and the text in the neighbourhood of that character to be convinced that the pattern does indeed match at that offset. To minimise the overall bound, set $k = \min\{ \sqrt{n}, m \}$.
\end{proof}


\subsection{A classical lower bound for random strings}

We finally complete the proof of Theorem \ref{thm:random} by giving the promised lower bound on the classical query complexity of pattern matching in random strings.

\begin{prop}
\label{prop:classrandomlb}
Fix $d=O(1)$. Any bounded-error classical algorithm which solves the pattern detection problem for random functions $P:[m]^d \rightarrow \Sigma$ and $T:[n]^d \rightarrow \Sigma$ must make at least $\Omega((n / m)^d\log_q m + n^{d/2} / \sqrt{\log_q n})$ queries in total, where $q = |\Sigma|$.
\end{prop}

\begin{proof}
An $\Omega((n \log_q m) / m)$ lower bound on the complexity of this problem was shown by Yao~\cite{yao79b} for the case $d=1$, which can be generalised to arbitrary $d$~\cite{karkkainen94}. To prove the second part of the bound, we use Lemma \ref{lem:injclasslb}. Assume for simplicity that $n$ and $m$ are each integer multiples of $b := (3 d \log_q n)^{1/d}$. Take an instance of the pattern detection problem and impose the constraint that the pattern is promised to match at an offset which is an integer multiple of $b$ in each dimension (this can only make the problem easier). Divide both the pattern and the text into blocks of $b \times \dots \times b$ characters, and concatenate the characters within each block to make a ``megacharacter'' of a larger alphabet $\Sigma'$ with $|\Sigma'| = q^{b^d} = n^{3d}$. If $P$ and $T$ are uniformly random, then each of the megacharacters will also be picked uniformly at random from $\Sigma'$. On the other hand, if $P$ matches $T$ at some offset, then each of the megacharacters will be random, except for those where $P$ matches $T$. Aside from these positions, all of the other characters in $T$ and $P$ will be unique except with probability at most $\binom{(n/b)^d+(m/b)^d}{2} n^{-3d} = O((nb^2)^{-d}) = O(n^{-d})$. Therefore, the $\Omega(n^{d/2})$ bound from Lemma \ref{lem:injclasslb} can be applied (replacing $n$ with $n / (3 d \log_q n)^{1/d}$), making minor modifications to account for the fact that the alphabet here is slightly bigger, which can only make the problem harder.
\end{proof}

This bound is tight up to poly-logarithmic factors, as by Lemma \ref{lem:randominj} a random pattern and text can be made injective at the cost of at most an $O(\log n)$ multiplicative factor, after which Theorem \ref{thm:classub} can be applied.


\section{Quantum algorithm for shift finding in $d$ dimensions}
\label{sec:dihedral}

In this section we describe a quantum algorithm for identifying hidden shifts in $d$ dimensions, eventually proving the following theorem:

\begin{repthm}{thm:dmatchapprox}
Let $d=O(1)$ and let $X$ be an arbitrary finite set. Let $f:\Z_{2^n}^d \rightarrow X$ and $g:\Z_{2^n}^d \rightarrow X$ be injective functions such that $\Pr_x[g(x) \neq f(x+s)] = O(n^{-2} 2^{-\sqrt{(2\log_2 3)dn}})$ for some $s \in \Z_{2^n}^d$. Then there is a quantum algorithm which outputs $s$ with bounded error using $O(n 2^{\sqrt{(2\log_2 3)dn}}) = O(n 2^{1.781\dots\sqrt{dn}})$ queries to each of these functions. The runtime is the same up to a $\poly(n)$ factor.
\end{repthm}

We first consider the case where $g$ exactly matches $f$ at some shift $s$. The algorithm to determine $s$ can be seen as a hybrid of two algorithms of Kuperberg~\cite{kuperberg05}, so we begin by discussing the intuition behind Kuperberg's original algorithm for the case $d=1$ (see~\cite{kuperberg05,regev04} for more). It is based on producing a number of states
\[ \ket{\psi_r} := \frac{1}{\sqrt{2}} \left( \ket{0} + \omega^{rs} \ket{1} \right) \]
for uniformly random $r \in \Z_{2^n}$, where we define $\omega := e^{\pi i/2^{n-1}}$. We describe later on how this can be done. The algorithm uses a combination operation which takes as input two states $\ket{\psi_r}$, $\ket{\psi_t}$. This operation returns $\ket{\psi_{r-t}}$ with probability $1/2$ (we call this ``success''), and $\ket{\psi_{r+t}}$ with probability $1/2$ (we call this ``failure''). The intention is to produce the state $\ket{\psi_{2^{n-1}}} = \frac{1}{\sqrt{2}} \left( \ket{0} + (-1)^{s_n} \ket{1} \right)$, where $s_n$ is the lowest-order bit of $s$. Given this state, the bit $s_n$ can be determined by applying a Hadamard gate and measuring; using this as a subroutine turns out to be sufficient to identify the hidden shift $s$ in its entirety.

The idea of Kuperberg~\cite{kuperberg05} which will enable us to produce such a state $\ket{\psi_{2^{n-1}}}$ quite efficiently is as follows. Divide the $n$ bits into $O(\sqrt{n})$ blocks of consecutive bits, starting with the lowest-order bits. Each block is of length $O(\sqrt{n})$, aside from the last block, which only contains the highest-order bit. The algorithm is split into a number of stages. The $i$'th stage of the algorithm is given as input a pool of states $\ket{\psi_r}$ such that all the bits of $r$ in the first $i$ blocks are zero. Then each state $\ket{\psi_r}$ is paired up with another state $\ket{\psi_{r'}}$ such that $r = r'$ on the bits in block $i+1$, if such a state exists. The combination operation is applied to these pairs to produce some new states $\ket{\psi_{r''}}$ such that all the bits in the $(i+1)$'st block of $r''$ are also 0. We repeat this process until we have zeroed all the bits, except the highest-order bit.

The probability that two random bit-strings agree in their lowest $O(\sqrt{n})$ bits is $2^{-O(\sqrt{n})}$. Thus it requires about $2^{O(\sqrt{n})}$ states to obtain ``many'' pairs whose lowest $O(\sqrt{n})$ bits are equal. Roughly a $1/2$ fraction of these pairs will successfully combine to give roughly a $1/4$ fraction of bit-strings which have their lowest-order bits zero, and can be input to the next stage. After applying the combination operation, the higher-order bits are still uniformly distributed, so we can repeat this argument. The net result is that we need to start with $2^{O(\sqrt{n})}$ states in total to have a good chance of eventually producing many states $\ket{\psi_{2^{n-1}}}$.

A similar idea can be used for the case $d>1$. Here the hidden shift $s \in \Z_{2^n}^d$ is thought of as a $d$-tuple $(s^{(1)},\dots,s^{(d)})$ of $n$-bit strings, and we redefine
\[ \ket{\psi_r} := \frac{1}{\sqrt{2}} \left( \ket{0} + \omega^{r\cdot s} \ket{1} \right), \]
where $r = (r^{(1)},\dots,r^{(d)})$ is a $d$-tuple of $n$-bit strings, and $r\cdot s = \sum_{i=1}^d r^{(i)} s^{(i)}$. We seek to learn the lowest-order bit $s^{(i)}_n$ of each bit-string. We achieve this by producing states $\ket{\psi_r}$ such that every bit of $r^{(i)}$ is zero for all $i$, except the highest-order bits $r^{(i)}_1$, which are uniformly random. Thus $\ket{\psi_r} = \frac{1}{\sqrt{2}} \left( \ket{0} + (-1)^{\sum_i r^{(i)}_1 s^{(i)}_n} \ket{1} \right)$. If we measure in the Hadamard basis, we learn the sum modulo 2 of a random subset of the $s^{(i)}_n$ values; repeating this $O(d)$ times is sufficient to learn all the $d$ bits $s^{(1)}_n,\dots,s^{(d)}_n$. In order to produce states $\ket{\psi_r}$ of this form, the same process as sketched above can be used to zero corresponding blocks of bits in every element of the $d$-tuple at once. It turns out to be more efficient to reduce the block size to $O(\sqrt{n/d})$; with this block size we end up with needing an initial pool of $2^{O(\sqrt{dn})}$ states.

The algorithm we give here achieves an improved complexity over Kuperberg's algorithm by noticing that the states $\ket{\psi_{r+t}}$ resulting from a failed combination operation between $\ket{\psi_r}$, $\ket{\psi_t}$ can be reused. If the pairs $(r,t)$ and $(r',t')$ all had the same low-order bits, so do the pair $(r+t,r'+t')$; and the high-order bits of $r+t$ and $r'+t'$ are still uniformly distributed. Assuming that we have $N \gg 2^{O(\sqrt{n})}$ states input to a given stage, we  expect roughly $N/4 + N/16 + \dots = N/3$ states to be put through to the next stage, an improvement over the previous $N/4$. An additional improvement is found by modifying the block size depending on the algorithm's progress. In early stages, the algorithm has access to a very large pool of states, so should try to zero many bits at once. Later on, there are fewer states available, so fewer bits are zeroed.

{\bf Other algorithms. }Kuperberg describes a second algorithm which is based on greedily choosing states $\ket{\psi_r}$, $\ket{\psi_t}$ to combine, in order to maximise the number of zero bits obtained~\cite{kuperberg05}. The running time of his algorithm is essentially the same as the algorithm described here. However, the analysis of the algorithm given here seems (to the author) somewhat easier; and in particular, here we give a full proof that the algorithm succeeds with high probability, which is omitted in~\cite{kuperberg05}. Kuperberg also outlines a general algorithm which works for any abelian group, rather than just the case $\Z_{2^n}^d$ we consider here. However, the complexity of this algorithm is not calculated precisely, only being given in the form $2^{O(\sqrt{n})}$. A subsequent algorithm of Regev~\cite{regev04} achieves improved space complexity over the algorithms of~\cite{kuperberg05}, but at the expense of increased time complexity. Another interesting algorithm for this problem, which achieves improved quantum (but not classical) space complexity, was recently given by Kuperberg~\cite{kuperberg13}. This algorithm is believed to have a somewhat faster runtime compared with the algorithm given here, but only a heuristic argument for this is currently known~\cite{kuperberg13}.


\subsection{The algorithm}

We now describe a quantum algorithm for solving the $d$-dimensional hidden shift problem. The algorithm is defined in terms of integers $N$ and $S$, and a list of integers $b_1,\dots,b_S$ such that $\sum_{i=1}^S b_i = n-1$. $N$ is the number of states used, $S$ is the number of stages of the algorithm, and $b_i$ is the number of bits which are zeroed during stage $i$. At each stage $i$ the algorithm operates on a list $\mathcal{L}_i$ of states, which are partitioned into bins.

The algorithm proceeds as follows:

\begin{enumerate}
\item Create a list $\mathcal{L}_1$ of $N$ states $\ket{\psi_r}$, for random $r \in \Z_{2^n}^d$.
\item Repeat the following operations for $i=1,\dots,S$:
\begin{enumerate}
\item Sort each state $\ket{\psi_r} \in \mathcal{L}_i$ into one of $2^{db_i}$ bins according to the values of the bits of $r$ at indices $n + 1 - \sum_{j=1}^i b_j,\dots,n - \sum_{j=1}^{i-1} b_j$.
\item Repeatedly perform the following steps, until the total number of states in all the bins is at most $n^2$:
\begin{enumerate}
\item Divide the states in each bin into pairs, discarding any left-over states.
\item Apply the combination operation to each pair.
\item For each successful operation, add the resulting state to $\mathcal{L}_{i+1}$. For each failure, leave the resulting state in the same bin as before.
\end{enumerate}
\end{enumerate}
\item The result is a list $\mathcal{L}_{S+1}$ of states of the form $\ket{\psi_r}$, where $r$ is uniformly distributed in $\{0,2^{n-1}\}^d$.
\end{enumerate}

Observe that the time complexity of the algorithm is $O(SN \log N)$, assuming that each element of the initial list can be created in time $O(1)$ and the combination operation takes time $O(1)$. This is because the time complexity of each stage is dominated by the sorting operation in step 2a, which can be carried out in time $O(N \log N)$. We have the following claim:

\begin{claim}
\label{claim:alg}
For arbitrary $k = O(1)$, there exist $N = O(n 2^{\sqrt{(2\log_2 3)dn}})$, $S = O(\sqrt{n})$ and a choice of integers $b_i$ such that the final list $\mathcal{L}_{S+1}$ satisfies $|\mathcal{L}_{S+1}| \ge k$ with probability at least $2/3$.
\end{claim}

We defer the proof of Claim \ref{claim:alg} to Appendix \ref{app:alg}. Assuming the correctness of this claim, we fill in the rest of the details about how the above algorithm can be used to solve the hidden shift problem, first showing how the states $\ket{\psi_r}$ can be produced. We begin by creating a number of quantum states of the form
\[ \frac{1}{\sqrt{2^{dn+1}}} \sum_{x \in \Z_{2^n}^d} \ket{x} \left( \ket{0}\ket{f(x)} + \ket{1}\ket{g(x)}\right). \]
Such a state can be produced using one query to each of $f$ and $g$. The last register is now measured and the remaining two registers are swapped (for notational clarity). Assuming the outcome $z \in X$ was received, the residual state will be
\[ \frac{1}{\sqrt{2}} \left(\ket{0}\ket{f^{-1}(z)} + \ket{1}\ket{g^{-1}(z)} \right) = \frac{1}{\sqrt{2}} \left(\ket{0}\ket{a} + \ket{1}\ket{a-s} \right), \]
for some element $a = f^{-1}(z)$, using injectivity of $f$ and $g$ and the fact that $g(x) = f(x+s)$. The next step is to perform the inverse quantum Fourier transform over $\Z_{2^n}^d$ on the second register. This is equivalent to performing the tensor product of $d$ inverse QFTs over $\Z_{2^n}$. The result is the state
\[ \frac{1}{\sqrt{2^{dn+1}}} \left( \sum_{x \in \Z_{2^n}^d} \omega^{-x \cdot a} \ket{0}\ket{x} + \sum_{x \in \Z_{2^n}^d} \omega^{-x \cdot (a-s)} \ket{1}\ket{x} \right) = \frac{1}{\sqrt{2^{dn+1}}} \sum_{x \in \Z_{2^n}^d} \left( \ket{0} + \omega^{x \cdot s} \ket{1} \right) \omega^{-x \cdot a} \ket{x}, \]
where as before $\omega := e^{\pi i/2^{n-1}}$ and $x \cdot a = \sum_{i=1}^d x^{(i)} a^{(i)}$, with multiplication taken over $\Z_{2^n}$. We now measure the second register and are left with the residual state
\[ \ket{\psi_r} := \frac{1}{\sqrt{2}} \left( \ket{0} + \omega^{r \cdot s} \ket{1} \right) \]
for some $r \in \Z_{2^n}^d$. Note that $r$ is uniformly random and we know what it is. We now describe how the combination operation works. Recall that, given two states $\ket{\psi_r}$, $\ket{\psi_t}$, this operation should produce $\ket{\psi_{r-t}}$ with probability $1/2$, and otherwise produce $\ket{\psi_{r+t}}$, where addition and subtraction are over $\Z_{2^n}^d$. This is simply achieved by measuring the parity of the two qubits; with probability $1/2$ we get ``even'' and the state collapses to $\frac{1}{\sqrt{2}}(\ket{00} + \omega^{(r+t)\cdot s}\ket{11})$, and with probability $1/2$ we get ``odd'' and the state collapses to $\frac{1}{\sqrt{2}}(\omega^{r \cdot s} \ket{10} + \omega^{t \cdot s}\ket{01})$. By relabelling basis states, and ignoring an overall phase factor, these are equivalent to $\ket{\psi_{r+t}}$ and $\ket{\psi_{r-t}}$ respectively.

Using the above algorithm, for arbitrary $k = O(1)$ we can produce a list of $k$ states of the form $\ket{\psi_r}$, where $r$ is uniformly distributed in $\{0,2^{n-1}\}^d$. Defining $\beta \in \{0,1\}^d$ by $\beta_i = r^{(i)}_n$, we have $\ket{\psi_r} = \frac{1}{\sqrt{2}}(\ket{0} + (-1)^{\beta \cdot s'}\ket{1} )$, where $s' \in \{0,1\}^d$ is the vector of lowest-order bits of $s$ and $\beta \cdot s' = \sum_{i=1}^d \beta_i s'_i$. Therefore, if we apply a Hadamard gate and measure, we learn the value of $\beta \cdot s'$ with certainty, for a random $\beta \in \{0,1\}^d$ (which we know). It suffices to do this $O(d)$ times to learn $s'$ completely with high probability. Picking $k = O(d)$, one run of the algorithm is enough to learn $s'$ with bounded failure probability.

Once the lowest-order bit of each component of $s$ has been learned, the remaining bits can be learned using the following idea. Let $\beta \in \{0,1\}^d$ be the lowest-order bits, and let $a \in \{0,1\}^d$ be arbitrary. Define $f'(x) = f(2x + a)$ for $x \in \Z_{2^{n-1}}^d$, and $g'(x) = g( 2x + a - \beta)$. Also let $s' \in \Z_{2^n}^d$ be the $d$-tuple obtained by setting $(s')^{(i)} = \lfloor s^{(i)}/2 \rfloor$. Then we have
\[ g'(x) = g(2x+a-\beta) = f(2x + a - \beta + 2s' + \beta) = f(2(x + s') + a) = f'(x+s'). \]
We can therefore apply the above algorithm to $f'$ and $g'$ to find the lowest-order bits of $s'$, repeating this procedure to learn $s$ completely.


\subsection{The approximate case}
\label{sec:dshiftapprox}

We now show that the above algorithm still works when $g$ only approximately matches $f$. The algorithm for obtaining the lowest-order bits of the shift $s$ is based on producing $O(n 2^{\sqrt{(2\log_2 3)dn}})$ quantum states of the form
\[ \ket{\phi_r} := \frac{1}{\sqrt{2}}\left(\ket{0}\ket{f(r)} + \ket{1}\ket{g(r)} \right) \]
for uniformly random $r \in \Z_{2^n}^d$, and applying some processing to these states. One can thus see the algorithm as operating on $O(n 2^{\sqrt{(2\log_2 3)dn}})$ copies of a mixed state
\[ \rho := \frac{1}{2^{nd}} \sum_{r \in \Z_{2^n}^d} \proj{r} \otimes \proj{\phi_r}. \]
If $f$ and $g$ are modified to take different values at up to $\epsilon 2^{nd}$ positions each, and we let $\widetilde{\rho}$ denote the corresponding mixed state, we have $\|\widetilde{\rho} - \rho\|_1 \le 2\epsilon$. This implies that, if $f$ and $g$ are modified by changing their values at up to an arbitrary $O(\delta n^{-1} 2^{-\sqrt{(2\log_2 3)dn}})$ fraction of positions, for arbitrary $0 \le \delta < 1$, with probability at least $1-\delta$ the algorithm does not notice the difference between having copies of $\rho$ and copies of $\widetilde{\rho}$ and will still output the lowest-order bits of the (necessarily unique) shift $s$ maximising $|\{x:g(x) = f(x+s)\}|$.

In each subsequent iteration, the algorithm solves the same problem for functions whose domain is half the size of the previous iteration, based on the previous functions evaluated at positions determined by an arbitrary offset $a \in \{0,1\}^d$. Consider the first such choice and let $f_a,g_a:\Z_{2^{n-1}}^d \rightarrow X$ denote the new functions. If we pick $a$ uniformly at random, we have
\[ \E_a[\Pr_x[f_a(x) \neq g_a(x)]] = \epsilon. \]
Therefore, by Markov's inequality $\Pr_a[ \Pr_x[f_a(x) \neq g_a(x)] \ge \alpha \epsilon] \le 1/\alpha$ for any $\alpha \ge 1$. The same argument holds for all subsequent iterations of the algorithm too. Using a union bound over the $O(n)$ iterations, taking $\alpha = O(n)$ and $\epsilon = O(n^{-2}2^{-\sqrt{(2\log_2 3)dn}})$ suffices to ensure that the algorithm succeeds with probability $\Omega(1)$.


\section{Outlook}
\label{sec:outlook}

We have described a quantum algorithm which achieves a super-polynomial separation from classical computation for the basic problem of pattern matching on average-case inputs. There are some undesirable aspects of our algorithm. First, the super-polynomial speedup is only achieved if we accept that for certain inputs, the algorithm might fail with high probability. This cannot be avoided: even checking whether a claimed match is really a match must take time $\Omega(m^{d/2})$ in the worst case. However, observe that the algorithm of Theorem \ref{thm:main} does have the property that, if it outputs an offset at which the  pattern is claimed to match the text, we can be confident that this only differs from a real match at an $O(\gamma)$ fraction of positions.

Second, the $2^{O(d^{3/2}\sqrt{\log m})}$ component of the runtime, while $o(m^{\epsilon})$ for any $\epsilon>0$, is still undesirably high. Substantially improving this term (to be logarithmic in $m$, for example) would presumably require finding an efficient quantum algorithm for the dihedral hidden subgroup problem, which has been a major open problem for over a decade. However, as the algorithm for finding hidden shifts is used as a black box, any improvements to this would imply an improved pattern-matching algorithm.

Finally, an interesting open question is whether efficient quantum algorithms can be found for approximate pattern matching. In the classical literature, Chang and Lawler describe an approximate pattern matching algorithm running in time $O((n/m) \log m)$ on random inputs, if one ignores the time to preprocess the pattern~\cite{chang94}. Another example is the classical work of Andoni et al.~\cite{andoni13} on the noisy hidden shift problem. The quantum algorithm described in Section \ref{sec:dihedral} can be used to solve the noisy hidden shift problem for random inputs as long as the noise rate is very low ($2^{-O(\sqrt{\log n})}$ for a problem of size $n$). Solving the hidden shift problem with a constant noise rate more efficiently than is possible classically seems likely to require new ideas.


\subsection*{Acknowledgments}

I would like to thank Rapha\"el Clifford, Markus Jalsenius and Ben Sach for helpful discussions on the topic of this paper. I would also like to thank one pseudonymous and two anonymous referees for a number of detailed comments and suggestions which have significantly improved the paper. This work was supported by an EPSRC Early Career Fellowship (EP/L021005/1).


\appendix

\section{Proof of Claim \ref{claim:alg}}
\label{app:alg}

In this appendix, we analyse the algorithm of Section \ref{sec:dihedral} for identifying hidden shifts. Let $L_i = |\mathcal{L}_i|$ denote the number of states in the list at the start of the $i$'th stage. Our first task is to find a lower bound on $L_{i+1}$ in terms of $L_i$ which holds with high probability. For ease of analysis, we consider a different, ``relaxed'' process where any state is allowed to be combined with any other state, rather than being restricted to the same bin; and where we assume there are always an even number of states. We then relate the relaxed process to the real process followed by the algorithm by bounding the number of states that the real process would need to discard, given its restriction to only combining states in the same bin.

Let $L_i^{(j)}$ denote the number of states we would have in the list $\mathcal{L}_i$ after $j$ steps of the $i$'th stage, based on following the relaxed process; note that $L_i^{(0)} = L_i$. Also let $S_i^{(j)}$ denote the number of successful pairings in the $j$'th step of the $i$'th stage (again allowing any pair of states to be combined). Finally let $t_i$ denote the number of steps taken in the $i$'th stage (i.e.\ until the total number of states is at most $n^2$). Then
\be \label{eq:baseineqs} L_i^{(j)} = \frac{L_i^{(j-1)}}{2} - S_i^{(j)},\;\;\;\; L_{i+1} \ge \sum_{j=1}^{t_i} S_i^{(j)} - t_i 2^{db_i}, \ee
where the second inequality allows for the fact that at each step we may need to discard at most $2^{db_i}$ states in the real process, as compared with the relaxed process. The probability that each combination operation applied to a pair succeeds is independent and equal to $1/2$. Using a Chernoff bound, we have
\be\label{eq:badevent} \Pr[ |S_i^{(j)} - L_i^{(j-1)}/4| \ge \ln n \sqrt{L_i^{(j-1)}/4} ] \le 2 e^{-(\ln n)^2 / 2} = 2n^{-(\ln n)/2}. \ee
As there will be $\poly(n)$ steps in total throughout the algorithm, this quantity is small enough that we can take a union bound over all steps and assume that this event never happens. Making this assumption, we have
\be \label{eq:listlb} L_i^{(j+1)} \ge \frac{L_i^{(j)}}{2} - \frac{L_i^{(j)}}{4} - \frac{\ln n}{2} \sqrt{L_i^{(j)}} = \frac{L_i^{(j)}}{4}\left(1 - \frac{2\ln n}{\sqrt{L_i^{(j)}}} \right). \ee
As we have $L_i^{(j)} \ge n^2$ for all steps $j$, then
\[ L_i^{(j+1)} \ge \frac{L_i^{(j)}}{4}\left(1 - \frac{2\ln n}{n} \right), \]
implying
\be \label{eq:listlbj} L_i^{(j)} \ge L_i \left(\frac{1 - (2\ln n)/n}{4} \right)^j. \ee
By (\ref{eq:baseineqs}), (\ref{eq:badevent}) and (\ref{eq:listlb}),
\[ L_{i+1} \ge \sum_{j=1}^{t_i} S_i^{(j)} - t_i 2^{db_i} \ge \sum_{j=1}^{t_i} \frac{L_i^{(j-1)}}{4}\left(1 - \frac{2\ln n}{\sqrt{L_i^{(j-1)}}}\right) - t_i 2^{db_i} \ge \frac{1}{4} \sum_{j=1}^{t_i} L_i^{(j-1)} \left(1 - \frac{2\ln n}{n}\right) - t_i 2^{db_i} \]
and so by (\ref{eq:listlbj}), writing $\xi = (2\ln n)/n$,
\beas
L_{i+1} &\ge& \frac{L_i}{4} (1-\xi) \sum_{j=0}^{t_i-1}  \left(\frac{1 - \xi}{4} \right)^j - t_i 2^{db_i} = \frac{L_i}{4} (1 - \xi) \frac{1 - ((1-\xi)/4)^{t_i} }{1 - (1-\xi)/4} - t_i 2^{db_i}\\
 &=& L_i\frac{(1 - \xi) (1 - ((1-\xi)/4)^{t_i}) }{3 + \xi} - t_i 2^{db_i}.
\eeas
A rough lower bound that follows from (\ref{eq:listlb}) for large enough $n$ is that $L_i^{(j+1)} \ge L_i^{(j)}/8$; and as $L_i^{(j+1)} \le L_i^{(j)}/2$ always, we have $1/8 \le L_i^{(j+1)}/L_i^{(j)} \le 1/2$. Using $L_i^{(0)} = L_i$ and $L_i^{(t_i)} \le n^2$ we obtain $\log_2 (L_i / n^2) \le t_i \le 3\log_2 (L_i / n^2)$. So
\[ ((1-\xi)/4)^{t_i} \le 4^{-t_i} \le n^4 / L_i^2. \]
Assuming that $L_i \ge n^3$ for all $1 \le i \le S$, we have
\[ \frac{(1 - \xi) (1 - ((1-\xi)/4)^{t_i}) }{3+\xi} \ge \frac{(1- (2\ln n) / n)(1-1/n^2)}{3+ (2\ln n) / n} = \frac{1- O((\log n)/n)}{3}. \]
Further assume that $L_i = 2^{O(\sqrt{n})}$ for all $i$, implying $t_i = O(\sqrt{n})$. Then
\[ L_{i+1} \ge \frac{(1- O((\log n)/n))L_i}{3} - O(\sqrt{n}) 2^{db_i}. \]
We now need to determine how large $L_1$ needs to be such that $L_{S+1}$ is still quite large. Working backwards, we can take
\[ L_{i-1} \le 3(1+O((\log n)/n))L_i + O(\sqrt{n})2^{db_{i-1}}. \]
Thus
\[ L_1 \le \left(3(1+O((\log n)/n)) \right)^S L_{S+1} + \left(3(1+O((\log n)/n)) \right)^{S-1} O(\sqrt{n})2^{db_S} + \dots + O(\sqrt{n})2^{db_1}. \]
Assuming that $L_{S+1} = O(1)$, and $S=O(\sqrt{n})$, we have
\[ L_1 = O\left(\sqrt{n} \sum_{i=1}^S 3^i 2^{db_i}\right) \]
We now need to pick values $S$, $b_i$ for the number of stages and the number of bits zeroed at each stage, such that $\sum_{i=1}^S b_i = n-1$, to minimise $\sum_{i=1}^S 3^i 2^{d b_i}$. We choose these values to make all the above terms equal to $2^{c \sqrt{n}}$ for some fixed $c$, i.e. $b_i = (c\sqrt{n} - (\log_2 3)i)/d$ (for simplicity ignoring the fact that $b_i$ has to be rounded to an integer). Relaxing to the constraint $\sum_i b_i = n$ for simplicity, we obtain $S c \sqrt{n} - (\log_2 3)S(S+1)/2 = dn$. Hence $c = d\sqrt{n}/S + (\log_2 3)(S+1)/(2\sqrt{n})$. Minimising this over $S$, we get that the minimum is found at $S = \sqrt{2 \log_3 2} \sqrt{dn}$, giving
\[ c = \frac{\sqrt{d}}{\sqrt{2 \log_3 2}} + \frac{\log_2 3(\sqrt{2 \log_3 2} \sqrt{dn} + 1)}{2 \sqrt{n}} =  \sqrt{(2\log_2 3)d} + O(1/\sqrt{n}).\]
Thus
\[ L_1 = O(\sqrt{n}\,S \, 2^{(\sqrt{2(\log_2 3)d}+O(1/\sqrt{n}))\sqrt{n}}) = O(n 2^{\sqrt{(2\log_2 3)dn}}) = O(n 2^{1.781\dots \sqrt{dn}}) \]
as claimed.

\bibliographystyle{plain}


\begin{thebibliography}{10}

\bibitem{ambainis02}
A.~Ambainis.
\newblock Quantum lower bounds by quantum arguments.
\newblock {\em Journal of Computer and System Sciences}, 64:750--767, 2002.
\newblock {\tt quant-ph/0002066}.

\bibitem{andoni13}
A.~Andoni, H.~Hassanieh, P.~Indyk, and D.~Katabi.
\newblock Shift finding in sub-linear time.
\newblock In {\em Proc. 24\textsuperscript{th} ACM-SIAM Symposium on Discrete
  Algorithms}, pages 457--465, 2013.

\bibitem{batu03}
T.~Batu, F.~Erg{\"u}n, J.~Kilian, A.~Magen, S.~Raskhodnikova, R.~Rubinfeld, and
  R.~Sami.
\newblock A sublinear algorithm for weakly approximating edit distance.
\newblock In {\em Proc. 35\textsuperscript{th} Annual ACM Symposium on Theory of
  Computing}, pages 316--324, 2003.

\bibitem{bennett97}
C.~Bennett, E.~Bernstein, G.~Brassard, and U.~Vazirani.
\newblock Strengths and weaknesses of quantum computing.
\newblock {\em SIAM Journal on Computing}, 26(5):1510--1523, 1997.
\newblock {\tt quant-ph/9701001}.

\bibitem{boyer77}
R.~Boyer and S.~Moore.
\newblock A fast string searching algorithm.
\newblock {\em Communications of the ACM}, 20(10):762--772, 1977.

\bibitem{brassard02}
G.~Brassard, P.~H{\o }yer, M.~Mosca, and A.~Tapp.
\newblock Quantum amplitude amplification and estimation.
\newblock {\em Quantum Computation and Quantum Information: A Millennium
  Volume}, pages 53--74, 2002.
\newblock {\tt quant-ph/0005055}.

\bibitem{chang94}
W.~Chang and E.~Lawler.
\newblock Sublinear approximate string matching and biological applications.
\newblock {\em Algorithmica}, 12(4--5):327--344, 1994.

\bibitem{childs07b}
A.~Childs and W.~van Dam.
\newblock Quantum algorithm for a generalized hidden shift problem.
\newblock In {\em Proc. 18\textsuperscript{th} ACM-SIAM Symposium on Discrete
  Algorithms}, pages 1225--1232, 2007.
\newblock {\tt quant-ph/0507190}.

\bibitem{childs13b}
A.~Childs, R.~Kothari, M.~Ozols, and M.~Roetteler.
\newblock Easy and hard functions for the boolean hidden shift problem.
\newblock In {\em Proc. 8\textsuperscript{th} Conference on the Theory of
  Quantum Computation, Communication, and Cryptography (TQC'13)}, pages 50--79,
  2013.
\newblock {\tt arXiv:1304.4642}.

\bibitem{crochemore07}
M.~Crochemore, C.~Hancart, and T.~Lecroq.
\newblock {\em Algorithms on Strings}.
\newblock Cambridge University Press, 2007.

\bibitem{curtis04}
D.~Curtis and D.~Meyer.
\newblock Towards quantum template matching.
\newblock In {\em Proc. SPIE 5161, Quantum Communications and Quantum Imaging},
  pages 134--141, 2004.

\bibitem{vandam06}
W.~van Dam, S.~Hallgren, and L.~Ip.
\newblock Quantum algorithms for some hidden shift problems.
\newblock {\em SIAM Journal on Computing}, 36:763--778, 2006.
\newblock {\tt quant-ph/0211140}.

\bibitem{ettinger04}
M.~Ettinger, P.~H{\o }yer, and E.~Knill.
\newblock The quantum query complexity of the hidden subgroup problem is
  polynomial.
\newblock {\em Information Processing Letters}, 91:43--48, 2004.
\newblock {\tt quant-ph/0401083}.

\bibitem{friedl03}
K.~Friedl, G.~Ivanyos, F.~Magniez, M.~Santha, and P.~Sen.
\newblock Hidden translation and orbit coset in quantum computing.
\newblock In {\em Proc. 35\textsuperscript{th} Annual ACM Symposium on Theory of
  Computing}, pages 1--9, 2003.
\newblock {\tt quant-ph/0211091}.

\bibitem{gavinsky11}
D.~Gavinsky, M.~Roetteler, and J.~Roland.
\newblock Quantum algorithm for the {B}oolean hidden shift problem.
\newblock In {\em Proc. 17\textsuperscript{th} International Computing \&
  Combinatorics Conference (COCOON'11)}, pages 158--167, 2011.
\newblock {\tt arXiv:1103.3017}.

\bibitem{gharibi13}
M.~Gharibi.
\newblock Reduction from non-injective hidden shift problem to injective hidden
  shift problem.
\newblock {\em Quantum Information and Computation}, 13(3{\&}4):221--230, 2013.
\newblock {\tt arXiv:1207.4537}.

\bibitem{giovannetti08}
V.~Giovannetti, S.~Lloyd, and L.~Maccone.
\newblock Quantum random access memory.
\newblock {\em Physical Review Letters}, 100:160501, 2008.
\newblock {\tt arXiv:0708.1879}.

\bibitem{grover97}
L.~Grover.
\newblock Quantum mechanics helps in searching for a needle in a haystack.
\newblock {\em Physical Review Letters}, 79(2):325--328, 1997.
\newblock {\tt quant-ph/9706033}.

\bibitem{hoyer03}
P.~H{\o }yer, M.~Mosca, and R.~de~Wolf.
\newblock Quantum search on bounded-error inputs.
\newblock In {\em Proc. 30\textsuperscript{th} {I}nternational {C}onference on
  {A}utomata, {L}anguages and {P}rogramming (ICALP'03)}, pages 291--299, 2003.
\newblock {\tt quant-ph/0304052}.

\bibitem{karkkainen94}
J.~K{\"a}rkk{\"a}inen and E.~Ukkonen.
\newblock Two- and higher-dimensional pattern matching in optimal expected time.
\newblock {\em SIAM Journal on Computing}, 29(2):571--589, 1999.

\bibitem{knuth77}
D.~Knuth, J.~Morris, Jr., and V.~Pratt.
\newblock Fast pattern matching in strings.
\newblock {\em SIAM Journal on Computing}, 6(2):323--350, 1977.

\bibitem{kuperberg05}
G.~Kuperberg.
\newblock A subexponential-time quantum algorithm for the dihedral hidden
  subgroup problem.
\newblock {\em SIAM Journal on Computing}, 35(1):170--188, 2005.
\newblock {\tt quant-ph/0302112}.

\bibitem{kuperberg13}
G.~Kuperberg.
\newblock Another subexponential-time quantum algorithm for the dihedral hidden
  subgroup problem.
\newblock In {\em Proc. 8\textsuperscript{th} Conference on the Theory of
  Quantum Computation, Communication, and Cryptography (TQC'13)}, pages 20--34,
  2013.
\newblock {\tt arXiv:1112.3333}.

\bibitem{montanaro13c}
A.~Montanaro and R.~de~Wolf.
\newblock Quantum property testing, 2013.
\newblock {\tt arXiv:1310.2035}.

\bibitem{moore07}
C.~Moore, D.~Rockmore, A.~Russell, and L.~Schulman.
\newblock The power of strong fourier sampling: Quantum algorithms for affine
  groups and hidden shifts.
\newblock {\em SIAM Journal on Computing}, 37(3):938--958, 2005.
\newblock {\tt quant-ph/0503095}.

\bibitem{navarro01}
G.~Navarro.
\newblock A guided tour to approximate string matching.
\newblock {\em ACM Computing Surveys}, 33(1):31--88, 2001.

\bibitem{ozols13}
M.~Ozols, M.~Roetteler, and J.~Roland.
\newblock Quantum rejection sampling.
\newblock {\em ACM Transactions on Computation Theory (TOCT)},
  5(3):11:1--11:33, 2013.
\newblock {\tt arXiv:1103.2774}.

\bibitem{ramesh03}
H.~Ramesh and V.~Vinay.
\newblock String matching in {$\widetilde{O}(\sqrt{n} + \sqrt{m})$} quantum
  time.
\newblock {\em Journal of Discrete Algorithms}, 1:103--110, 2003.
\newblock {\tt quant-ph/0011049}.

\bibitem{regev04}
O.~Regev.
\newblock A subexponential time algorithm for the dihedral hidden subgroup
  problem with polynomial space, 2004.
\newblock {\tt quant-ph/0406151}.

\bibitem{roetteler09}
M.~R{\"o}tteler.
\newblock Quantum algorithms to solve the hidden shift problem for quadratics
  and for functions of large {G}owers norm.
\newblock In {\em Proc. 39th International Symposium on Mathematical Foundations of Computer Science (MFCS'09)}, LNCS vol. 5734, pages 663--674, 2009.
\newblock {\tt arXiv:0911.4724}.

\bibitem{roetteler10}
M.~R{\"o}tteler.
\newblock Quantum algorithms for highly non-linear {B}oolean functions.
\newblock In {\em Proc. 21\textsuperscript{st} ACM-SIAM Symposium on Discrete
  Algorithms}, pages 448--457, 2010.
\newblock {\tt arXiv:0811.3208}.

\bibitem{twamley00}
J.~Twamley.
\newblock A hidden shift quantum algorithm.
\newblock {\em Journal of Physics A: Mathematical and General}, 33:8973, 2000.

\bibitem{vishkin91}
U.~Vishkin.
\newblock Deterministic sampling -- a new technique for fast pattern matching.
\newblock {\em SIAM Journal on Computing}, 20(1):22--40, 1991.

\bibitem{yao79b}
A.~Yao.
\newblock The complexity of pattern matching for a random string.
\newblock {\em SIAM Journal on Computing}, 8(3):368--387, 1979.

\end{thebibliography}

\end{document}